\documentclass[11pt]{article}

\usepackage[utf8]{inputenc}
\usepackage[english]{babel}

\usepackage{hyperref}
\usepackage{xcolor}
\usepackage{amsmath,amssymb}
\usepackage{enumerate}
\usepackage{amsthm}
\usepackage[all,2cell]{xy}
\usepackage{mathrsfs}
\usepackage{mathtools}
\usepackage{tikz}

\usepackage{geometry}
\def\baselinestretch{1.1}
\newtheorem{thm}{Theorem}[section]
\newtheorem{dfn}[thm]{Definition}
\newtheorem{prop}[thm]{Proposition}

\newtheorem{lem}[thm]{Lemma}
\newtheorem{exmpl}[thm]{Example}
\newtheorem{obs}[thm]{Remark}

\def\beq{\begin{equation}}
\def\eeq{\end{equation}}
\def\bea{\begin{eqnarray}}
\def\eea{\end{eqnarray}}
\def\beann{\begin{eqnarray*}}
\def\eeann{\end{eqnarray*}}
\def\ben{\begin{enumerate}}
\def\een{\end{enumerate}}
\def\bit{\begin{itemize}}
\def\eit{\end{itemize}}

\def\derpar#1#2{\frac{\partial{#1}}{\partial{#2}}}

\newcommand\restr[2]{{
  \left.\kern-\nulldelimiterspace 
  #1 
  \right|_{#2} 
}}

\newcommand{\R}{\mathbb{R}}

\renewcommand{\d}{\mathrm{d}}

\renewcommand{\L}{\mathcal{L}}
\renewcommand{\H}{\mathcal{H}}
\newcommand{\vf}{\mathfrak{X}}
\newcommand{\df}{\Omega}
\newcommand{\Cinfty}{\mathscr{C}^\infty}

\newcommand{\Tan}{\mathrm{T}}
\newcommand{\inn}{i}
\newcommand{\Lie}{\mathscr{L}}

\newcommand{\X}{\mathfrak{X}}
\newcommand{\Reeb}{\mathcal{R}}

\def\d{\mathrm{d}}

\let\ds\displaystyle

\title{\sc
A contact geometry framework\\
for field theories with dissipation}

\author{\sffamily 
$^a$Jordi Gaset, 
$^b$Xavier Gr\`acia, 
$^b$Miguel C. Mu\~noz-Lecanda,
$^b$Xavier Rivas and 
$^b$Narciso Rom\'an-Roy%
\thanks{emails: 
jordi.gaset@uab.cat,
xavier.gracia@upc.edu,
miguel.carlos.munoz@upc.edu,
xavier.rivas@upc.edu,
narciso.roman@upc.edu}
\\[1ex]
\normalsize\itshape\sffamily 
$^a$Department of Physics,
Universitat Aut\`onoma de Barcelona\\
\normalsize\itshape\sffamily 
Bellaterra, Catalonia, Spain
\\[0.1ex]
\normalsize\itshape\sffamily 
$^b$Department of Mathematics,
Universitat Polit\`ecnica de Catalunya\\
\normalsize\itshape\sffamily 
Barcelona, Catalonia, Spain
}

\date{\sffamily 
18 February 2020}

\begin{document}
\maketitle

\kern -10mm
\begin{abstract}
\noindent
We develop a new geometric framework suitable for dealing with
Hamiltonian field theories with dissipation.
To this end we define the notions of
$k$-contact structure
and $k$-contact Hamiltonian system.
This is a generalization of both
the contact Hamiltonian systems in mechanics and the
$k$-symplectic Hamiltonian systems in field theory.
The concepts of symmetries and dissipation laws are introduced and developed. 
Two relevant examples are analyzed in detail: 
the damped vibrating string and Burgers' equation.
\end{abstract}


\noindent\textbf{MSC\,2010 codes:}
53C15, 
53D10,
58A10,
70S05, 
70S10,
35R01



\medskip
\setcounter{tocdepth}{2}
{
\def\baselinestretch{0.97}
\small
\def\addvspace#1{\vskip 1pt}
\parskip 0pt plus 0.1mm
\tableofcontents
}

\section{Introduction}

In the last decades, the methods of geometric mechanics and field theory
have been widely used in order to give a geometrical description of a
large variety of systems in physics and applied mathematics;
in particular, those of symplectic and multisymplectic or $k$-symplectic (polysymplectic) geometry
(see, for instance, \cite{abraham1978,Arnold1989,Carinena1991,DeLeon2015,Libermann1987,Roman2011,Roman2009} and references therein).
All these methods are developed, in general, to model systems
of variational type; that is, without dissipation or damping, both
in the Lagrangian and Hamiltonian formalisms.

Furthermore, in recent years, there has been a growing interest in studying 
a geometric framework to describe dissipative or damped systems,
specifically using contact geometry 
\cite{Banyaga2016,Geiges2008,Kholodenko2013}.
The efforts have been focused in the
study of mechanical systems \cite{Bravetti2017b,Bravetti2017a,Ciaglia2018,DeLeon2019,DeLeon2016b,GGMRR-2019b,Lainz2018}.
All of them are described by ordinary differential equations
to which some terms that account for the dissipation or damping have been added.
Contact geometry has other physical applications,
as for instance thermodynamics 
\cite{Bravetti2018}.

Nevertheless, up to our knowledge,
the analysis of systems of partial differential equations with dissipation terms, that is in field theory,
has not yet been done geometrically.
Our basic geometrical model for classical field theories without dissipation is the $k$-symplectic framework
 \cite{DeLeon2015},
which is the simplest extension of the symplectic formulation of autonomous mechanics
to field theory.
In this way, the aim of this paper is to develop an extension of the contact geometry
in order to create a geometrical framework
to deal with these kinds of systems when dissipation is present.

As it is well-known, 
a simple contact structure can be defined 
starting from an exact symplectic manifold 
$(N,\omega = -\d\theta)$,
taking the product 
$M = N \times \R$,
and endowing this manifold with the contact form
$\eta = \d s - \theta$
(where $s$ is the cartesian coordinate of~$\R$).
So, 
the \emph{contactification} of the symplectic structure
is obtained by the addition of a contact variable~$s$ \cite[appendix~4]{Arnold1989}.
Given a Hamiltonian function $H \colon M \to \R$,
one defines the contact Hamilton equations,
analogous to the usual Hamilton equations
but with a dissipation term originating from the new variable.

Now we move to field theory, 
and more specifically to the De Donder--Weyl covariant formulation
of Hamiltonian field theory.
Aiming to introduce dissipation terms in the
Hamilton--De Donder--Weyl equations,
one realizes that we need to introduce 
$k$ contact variables $s^\alpha$,
where $k$ is the number of independent variables
(for instance the space-time dimension).
In the autonomous case,
the De Donder--Weyl formulation of Hamiltonian field theory
can be geometrically modeled with 
the notion of $k$-symplectic structure,
that is,
a family of $k$ differential 2-forms $\omega^\alpha$
satisfying some conditions.

These considerations lead us to define the concept of
$k$-contact structure on a manifold~$M$,
as a family of $k$ differential 1-forms $\eta^\alpha$ 
satisfying certain properties.
This structure implies the existence of two special
tangent distributions;
one of them spanned by $k$ Reeb vector fields
which will be instrumental 
in the formulation of field equations.
Then a $k$-contact Hamiltonian system
is defined as a $k$-contact manifold
endowed with a Hamiltonian function $\H$.
This structure allows us to define
$k$-contact Hamilton equations,
which are a generalization of 
the De Donder--Weyl Hamiltonian formalism,
and enables us to describe field theories with dissipation.
After that, we can study the symmetries for these Hamiltonian systems and, associated to them, the notion of dissipation law, which is characteristic of dissipative systems, and is analogous to the conservation laws of conservative systems.

The relevance of our framework 
is illustrated with two noteworthy examples:
the vibrating string with damping, 
and Burgers' equation.
The Lagrangian formulation of the 
(undamped) vibrating string is well known, 
and from its Hamiltonian counterpart
and a contactification procedure
we obtain its field equation with a damping term.
The case of Burgers' equation is more involved.
First we consider the heat equation;
although this equation is not variational,
the introduction of an auxiliary variable 
allows to describe the heat equation 
in Lagrangian terms.
From this we have provided a Hamiltonian field theory
that still describes the heat equation.
Finally, 
an appropriate contactification of this equation
yields Burgers' equation.
In both cases, there are different $k$-contact structures (with $k=2$), hidden in the standard treatment,
but uncovered in this geometric formulation.

The paper is organized as follows.
Section \ref{prel} is devoted to review briefly several preliminary concepts on
contact geometry and contact Hamiltonian systems in mechanics,
as well as on $k$-symplectic manifolds
and $k$-symplectic Hamiltonian systems in field theory.
In Section \ref{kcontact}
we introduce the definition of $k$-contact structure,
we give the basic definitions and properties of $k$-contact manifolds, and
we include a version of the Darboux theorem for a particular type of these manifolds.
In Section \ref{sect-hamsyst} we set a
geometric framework for Hamiltonian field theories 
with dissipation on a $k$-contact manifold
and state the geometric form of the contact Hamilton--De Donder--Weyl equations
in several equivalent ways.
Section \ref{examples} is devoted to study some relevant examples, 
in particular, the damped vibrating string
and the Burgers equation.
Finally, in Section \ref{sect-symms} 
we introduce two concepts of symmetry and
the relations between them, and the notions of
dissipation laws
for these kinds of systems.

Throughout the paper all the manifolds and mappings are assumed to be smooth. 
Sum over crossed repeated indices is understood.

\section{Preliminaries}
\label{prel}

In Section 3 we will introduce the notion of $k$-contact structure.
It is based on the notions of contact and $k$-symplectic structures, 
which we review in this section.

\subsection{Contact manifolds and contact Hamiltonian systems}

The geometry of contact manifolds is described in several books. We are interested, in particular, in contact Hamiltonian systems, 
see for instance
\cite{Bravetti2017b,Bravetti2018,Bravetti2017a,Ciaglia2018,DeLeon2019,Geiges2008,Lainz2018}.
			
\begin{dfn}\label{dfn-contact-manifold}
Let $M$ be a $(2n+1)$-dimensional  manifold. 
A \textbf{contact form} in $M$ is a differential $1$-form
$\eta\in\df^1(M)$ such that $\eta\wedge(\d\eta)^{\wedge n}$
is a volume form in $M$.
Then, $(M,\eta)$ is said to be a \textbf{contact manifold}.
\end{dfn}

\begin{obs}{\rm
Notice that the condition that
$\eta\wedge(\d\eta)^{\wedge n}$ is a volume form is equivalent to demand that
$$
\Tan M=\ker\,\eta\oplus\ker\,\d\eta\ .
$$}
\end{obs}

\begin{prop}
Let $(M,\eta)$ be a contact manifold. 
Then there exists a unique vector field $\Reeb\in\vf(M)$, 
which is called the \textbf{Reeb vector field}, 
such that
$$
	i({\Reeb})\eta =1 \,,\quad i({\Reeb})\d\eta = 0 \ .
$$
\end{prop}

The local structure of contact manifolds is given by the following theorem:

\begin{thm}[Darboux theorem for contact manifolds (\cite{Go-69} p.\,118)]
Let $(M,\eta)$ be a contact manifold. 
Then around each point of $M$ there exist an open set with local coordinates 
$(q^i,p_i, s)$ with $1\leq i\leq n$ 
such that
\begin{equation*}
	\eta = \d s- p_i\,\d q^i \  .
\end{equation*}
These are the so-called \textbf{Darboux} or \textbf{canonical coordinates} 
of the contact manifold $(M,\eta)$.
\end{thm}

In Darboux coordinates, the Reeb vector field is
$\displaystyle\Reeb = \frac{\partial}{\partial s}$.

The canonical model for contact manifolds is the manifold 
$\Tan^*Q\times\R$. 
In fact, if 
$\theta\in\Omega^1(\Tan^*Q)$ 
and 
$\omega=-\d\theta\in\Omega^2(\Tan^*Q)$
are the canonical forms in $\Tan^*Q$,
and 
$\pi_1 \colon \Tan^*Q \times \R \to \Tan^*Q$ 
is the canonical projection,
then 
$\eta=\d s-\pi_1^*\theta$ 
is a contact form in $\Tan^*Q\times\R$,
with 
$\d\eta=\pi_1^*\omega$.

Finally, given a contact manifold $(M,\eta)$, 
we have the following $\Cinfty(M)$-module isomorphism
\cite{Al-89}
\begin{equation*}
\begin{array}{rccl}
	\flat\colon & \X(M) & \longrightarrow & \Omega^1(M) 
	\\
	& X & \longmapsto & i(X)\d\eta-(i(X)\eta)\eta
\end{array}
\end{equation*}
and as a straightforward consequence we have:

\begin{thm}
If $(M,\eta)$ is a contact manifold, for every $H\in\Cinfty(M)$, 
there exists a unique vector field $X_H\in\vf(M)$ such that
\begin{equation}\label{Hvf}
    i(X_H)\d\eta=\d H-({\cal R}(H))\eta 
    \,, \quad
    i(X_H)\eta=-H \ .
\end{equation}
\end{thm}

\begin{dfn}
The vector field $X_H$ defined by \eqref{Hvf} is called the
\textbf{contact Hamiltonian vector field} associated with $H$
and equations \eqref{Hvf} are the \textbf{contact Hamilton equations}.
The triple $(M,\eta,H)$ is a \textbf{contact Hamiltonian system}.
\end{dfn}

\begin{obs}{\rm
Notice that the contact Hamiltonian equations are equivalent to 
$$
\Lie_{X_\H}\eta = - (\Lie_{\Reeb}\H)\eta 
\:,\quad
\inn(X_\H)\eta = - \H 
\:.
$$
Furthermore, the contact Hamiltonian vector field is such that
$\flat(X_H)=i(X_H)\d\eta-H\eta$.}
\end{obs}

Taking Darboux coordinates $(q^i, p_i, s)$, the contact Hamiltonian vector field is
$$ X = \frac{\partial\H}{\partial p_i}\frac{\partial}{\partial q^i} - \left(\frac{\partial\H}{\partial q^i} + p_i\frac{\partial\H}{\partial s}\right)\frac{\partial}{\partial p_i} + \left(p_i\frac{\partial\H}{\partial p_i} - \H\right)\frac{\partial}{\partial s}\ . $$
Hence, an integral curve of this vector field satisfies the contact Hamilton equations:
$$
\dot q^i = \frac{\partial\H}{\partial p_i}
\:,\quad
\dot p_i = -\left(\frac{\partial\H}{\partial q^i} + p_i\frac{\partial\H}{\partial s}\right)
\:,\quad
\dot s = p_i\frac{\partial\H}{\partial p_i} - \H
\:.
$$

\subsection{k-vector fields and integral sections}

The definition of $k$-vector field and 
its usage in the geometric study of partial differential equations
can be found in
\cite{DeLeon2015,Rey2004},
for instance.

Let $M$ be a manifold, and consider the direct sum of $k$ copies of its tangent bundle,
$\oplus^k \Tan M$.
It is endowed with the natural projections to each direct summand and to the base manifold:
$$
\tau^\alpha\colon \oplus^k \Tan M \rightarrow \Tan M
\,,\quad 
\tau^1_M  \colon \oplus^k \Tan M \to M 
\,.
$$

\begin{dfn}
\label{kvector}
A \textbf{$k$-vector field} on a manifold $M$ is a section
${\bf X} \colon M \to \oplus^k \Tan M$ 
of the projection~$\tau_M^1$.
\end{dfn}

A $k$-vector field ${\bf X}$ is specified by giving 
$k$ vector fields $X_{1}, \dots, X_{k}\in\vf(M)$, obtained as
$X_\alpha=\tau_\alpha\circ{\bf X}$.
Then it is denoted ${\bf X}=(X_1, \ldots, X_k)$.
A $k$-vector field ${\bf X}=(X_1, \ldots, X_k)$ 
induces a decomposable 
contravariant skew-symmetric tensor field,
$X_1 \wedge \ldots \wedge X_k$,
which is a section of $\Lambda^k \Tan M  \to M$.
This also induces a tangent distribution on~$M$.

\begin{dfn}
Given a map $\phi\colon D\subset\R^k \rightarrow M$,
the \textbf{first prolongation} 
of $\phi$ to $\oplus^k \Tan M$
is the map 
$\phi' \colon D \subset \R^k \to \oplus^k \Tan M$ 
defined by
$$
\phi'(t) =
\left(\phi(t),\Tan \phi \left(\frac{\partial}{\partial t^1}\Big\vert_{t}\right),
\ldots,
\Tan\phi\left(\frac{\partial}{\partial t^k}\Big\vert_{t}\right)\right)
\equiv (\phi(t);\phi_\alpha'(t))
\,,
$$
where $t = (t^1,\ldots,t^k)$ are the cartesian coordinates of $\R^k$.
\end{dfn}

\begin{dfn}
\label{integsect}
An \textbf{integral section} of a $k$-vector field
${\bf X}=(X_{1},\dots, X_{k})$ is a map
$\phi \colon D\subset\R^k \rightarrow M$ 
such that
$$
\phi' = {\bf X} \circ \phi
\,.
$$
Equivalently,
$\ds
\Tan \phi \circ \frac{\partial}{\partial t^\alpha}=X_\alpha \circ \phi
$
for every~$\alpha$.

A $k$-vector field ${\bf X}$ is \textbf{integrable} if
every point of~$M$ is in the image of an integral section of~${\bf X}$.
\end{dfn}
In coordinates, if
$\displaystyle X_\alpha= X_\alpha^i \frac{\partial}{\partial x^i}$,
then $\phi$ is an integral section of $\mathbf{X}$ if, and only if,
it is a solution to the following system of partial differential equations:
$$
\frac{\partial \phi^i}{\partial t^\alpha} = X_\alpha^i(\phi) \ .
$$
A $k$-vector field ${\bf X}=(X_1, \ldots, X_k)$ is integrable
if, and only if, $[X_\alpha,X_\beta] = 0$, for every $\alpha,\beta$ 
(see, for instance, \cite{Lee2013}),
and these are the necessary and sufficient conditions 
for the integrability of the above system of partial differential equations.

\subsection{\texorpdfstring{$k$}--symplectic manifolds and
\texorpdfstring{$k$}--symplectic Hamiltonian systems}

A simple geometric framework for Hamiltonian field theory 
can be built upon the notion of $k$-symplectic structure;
see for instance
\cite{Awane1992,DeLeon1988,DeLeon1988a,DeLeon2015,Rey2004}.

\begin{dfn}
\label{defaw}
Let $M$ be a manifold of dimension $N=n+kn$.
A \textbf{$k$-symplectic structure} on~$M$
is a family $(\omega^1,\ldots,\omega^k;V)$,
where $\omega^\alpha$ ($\alpha=1,\ldots,k$) are closed $2$-forms,
and $V$ is an integrable $nk$-dimensional tangent distribution on~$M$
such that
$$
(i) \ \omega^\alpha \vert_{V\times V} =0 \ \hbox{\rm (for every $\alpha$)}
\:, \quad
(ii)\ \cap_{\alpha=1}^{k} \ker\omega^\alpha = \{0\} 
\:.
$$
Then $(M,\omega^\alpha,V)$ is called a \textbf{$k$-symplectic manifold}.
\end{dfn}

If $(M,\omega^\alpha,V)$ is a $k$-symplectic manifold,
for every point of $M$ there exist a neighborhood $U$
and local coordinates
$(q^i , p^\alpha_i)$ ($1\leq i\leq n$, $1\leq \alpha\leq k$)
such that, on~$U$,
$$
\omega^\alpha=  \d q^i\wedge \d p^\alpha_i 
\:,\quad
V =
\left\langle \frac{\partial}{\partial p^1_i}, \dots,
\frac{\partial}{\partial p^k_i} \right\rangle 
\,.
$$
These are the so-called \emph{Darboux} or \emph{canonical coordinates}
of the $k$-symplectic manifold \cite{Awane1992}.

The canonical model for $k$-symplectic manifolds is
$\oplus^k \Tan^*Q= \Tan^*Q\oplus \stackrel{k}{\dots} \oplus \Tan^*Q$,
with natural projections
$$
\pi^\alpha\colon \oplus^k \Tan^*Q \rightarrow\Tan^*Q
\:, \quad 
\pi^1_Q \colon \oplus^k \Tan^*Q \to Q 
\:.
$$
As in the case of the cotangent bundle,
local coordinates $(q^i)$ in $U \subset Q$ induce
induced natural coordinates $(q^i ,p^\alpha_i)$ in
$(\pi^1_Q)^{-1}(U)$.

If $\theta$ and $\omega=-\d\theta$ are the canonical forms of $\Tan^*Q$,
then $\oplus^k \Tan^*Q$ is endowed with the canonical forms
$$
\theta^\alpha=(\pi^\alpha)^*\theta 
\,,\quad
\omega^\alpha=(\pi^\alpha)^*\omega=-(\pi^\alpha)^*\d\theta=-\d\theta^\alpha 
,
$$
and in natural coordinates we have that 
$\theta^\alpha = p^\alpha_i\d q^i$ 
and
$\omega^\alpha=\d q^i\wedge\d p^\alpha_i$.
Thus, the triple  
$(\oplus^k \Tan^*Q,\omega^\alpha, V)$,
where $V=\ker \Tan \pi^1_Q$,
is a $k$-symplectic manifold,
and the natural coordinates in $\oplus^k \Tan^*Q$ are Darboux coordinates.

\begin{dfn}
A \textbf{$k$-symplectic Hamiltonian system} is a family $(M,\omega^\alpha,V,H)$,
where $(M,\omega^\alpha,V)$ is a $k$-symplectic manifold,
and $H\in \Cinfty(M)$ is called a {\sl Hamiltonian function}.
\\
The \textbf{Hamilton--De Donder--Weyl equation} for a map
$\psi \colon D\subset\R^k \to M$ is
\begin{equation}
\inn({\psi'_\alpha}) \omega^\alpha = \d H  \circ \psi
\,.
\label{he20}
\end{equation}
The \textbf{Hamilton--De Donder--Weyl equation} for a $k$-vector field
${\bf X}=(X_1,\dots,X_k)$ in $M$ is
\begin{equation}\label{generic0}
\inn(X_\alpha)\,\omega^\alpha = \d H 
\,.
\end{equation}
\end{dfn}

For $k$-symplectic Hamiltonian systems,
solutions to this equation always exist, 
although they are not unique.
Moreover, solutions are nor necessarily integrable.

In canonical coordinates, if 
$\psi=(\psi^i,\psi^\alpha_i)$,
then 
$\ds 
\psi'_\beta  =
\Big(
\psi^i,\psi^\alpha_i,\derpar{\psi^i}{t^\beta},\derpar{\psi^\alpha_i}{t^\beta}
\Big)$,
and equation \eqref{he20} reads
$$
\frac{\partial\psi^i}{\partial t^\alpha} =
\frac{\partial H}{\partial p^\alpha_i}\circ\psi 
\:,\quad
\frac{\partial\psi^\alpha_i}{\partial t^\alpha} =
- \frac{\partial H}{\partial q^i}\circ\psi 
\:.
$$
Furthermore,
if ${\bf X}=(X_\alpha)$ is a $k$-vector field solution to \eqref{generic0} and
$\displaystyle
X_\alpha= (X_\alpha)^i\frac{\partial}{\partial q^i}+
(X_\alpha)_i^\beta\frac{\partial}{\partial p_i^\beta}
$, 
then
$$
\frac{\partial H}{\partial q^i} =  -(X_\alpha)^\alpha_i 
\:,\quad
\frac{\partial H}{\partial p^\alpha_i} = (X_\alpha)^i 
\:.
$$

\begin{prop}
{\rm \cite{DeLeon2015,Rey2004}}.
If ${\bf X}$ is an integrable $k$-vector field on~$M$ then
every integral section $\psi \colon D\subset\R^k \to M$ of ${\bf X}$
satisfies equation (\ref{he20})
if, and only if,
${\bf X}$ is a solution to equation \eqref{generic0}.
\end{prop}

Equations (\ref{he20}) and (\ref{generic0})
are not, in general, fully equivalent:
a solution to equation (\ref{he20})
may not be
an integral section of some integrable $k$-vector field in~$M$
solution to \eqref{generic0}.

\section{\texorpdfstring{$k$}--contact structures}
\label{kcontact}

Next we develop the general geometric framework
of our formalism.

\subsection{Definitions and basic properties}
\label{subsect-contact}

Let $M$ be a smooth manifold of dimension~$m$.
A (generalized) distribution on~$M$ 
is a subset $D \subset \Tan M$
such that, for every $x \in M$,
$D_x \subset \Tan_xM$ is a vector subspace.
$D$ is called smooth when it can be locally spanned
by a family of smooth vector fields;
it is called regular when it is smooth
and of locally constant rank.
One defines in the same way the notion of
codistribution, as a subset $C \subset \Tan^*M$.
The annihilator $D^\circ$ of a distribution~$D$ 
is a codistribution, 
but if $D$ is not regular then 
$D^\circ$ may not be smooth.
Within the usual identification $E^{**} = E$
of finite-dimensional linear algebra,
we have $(D^{\circ})^{\circ} = D$.

A (smooth) differential 1-form $\eta \in \Omega^1(M)$
generates a smooth codistribution that we denote by
$\langle \eta \rangle \subset \Tan^*M$;
it has rank~1 at every point where $\eta$ does not vanish.
Its annihilator is a distribution 
$\langle \eta \rangle^\circ \subset \Tan M$;
it can be described also as the kernel
of the linear morphism
$\widehat \eta \colon \Tan M \to M \times \R$
defined by~$\eta$.
This distribution has corank~1 
at every point where $\eta$ does not vanish.

In a similar way,
a differential 2-form $\omega \in \Omega^2(M)$
induces a linear morphism
$\widehat \omega \colon \Tan M \to \Tan^*M$,
$\widehat\omega(v) = \inn(v)\omega$.
Its kernel is a distribution
$\ker \widehat\omega \subset \Tan M$.
Recall that the rank of $\widehat \omega$ is an even number.

\medskip
Now we consider $k$ differential 1-forms
$\eta^1, \ldots, \eta^k \in \Omega^1(M)$,
and introduce the following notations:
\begin{itemize}
\item 
$\mathcal{C}^{\mathrm{C}} =
\langle \eta^1, \ldots, \eta^k \rangle \subset
\Tan^*M$;
\item 
$\mathcal{D}^{\mathrm{C}} =
\left( \mathcal{C}^{\mathrm{C}} \right)^\circ =
\ker \widehat{\eta^1} \cap \ldots \cap \ker \widehat{\eta^k} \subset
\Tan M$;
\item 
$\mathcal{D}^{\mathrm{R}} =
\ker \widehat{\d \eta^1} \cap \ldots \cap \ker \widehat{\d \eta^k} \subset
\Tan M$;
\item
$\mathcal{C}^{\mathrm{R}} = 
\left( \mathcal{D}^{\mathrm{R}} \right)^\circ\subset
\Tan^*M$.
\end{itemize}

\begin{dfn}
\label{kconman}
A \textbf{$k$-contact structure} on a manifold~$M$
is a family of $k$ differential 1-forms 
$\eta^\alpha \in \Omega^1(M)$
such that,
with the preceding notations,
\begin{enumerate}[(i)]
\item 
$\mathcal{D}^{\mathrm{C}} \subset \Tan M$
is a regular distribution of corank~$k$;
\item 
$\mathcal{D}^{\mathrm{R}} \subset \Tan M$
is a regular distribution of rank~$k$;
\item
$\mathcal{D}^{\mathrm{C}} \cap \mathcal{D}^{\mathrm{R}} = \{0\}$.
\end{enumerate}
We call
$\mathcal{C}^{\mathrm{C}}$
the \textbf{contact codistribution};
$\mathcal{D}^{\mathrm{C}}$ 
the \textbf{contact distribution};
$\mathcal{D}^{\mathrm{R}}$ 
the \textbf{Reeb distribution};
and
$\mathcal{C}^{\mathrm{R}}$
the \textbf{Reeb codistribution}.
\\
A \textbf{$k$-contact manifold} is a manifold endowed with a $k$-contact structure.
\end{dfn}

\begin{obs}\rm
Condition (i) in Definition \ref{kconman} is equivalent to each one of these two conditions:
\begin{enumerate}
\item[\it (i\,$'$)]
$\mathcal{C}^{\mathrm{C}} \subset \Tan^*M$
is a regular codistribution of rank~$k$;
\item[\it (i\,$''$)]
$\eta^1 \wedge \ldots \wedge \eta^k \neq 0$
at every point.
\end{enumerate}
Condition (iii) can be obviously rewritten as
\begin{enumerate}
\item[\it (iii\,$'$)]
$\displaystyle
\bigcap_{\alpha=1}^{k} \left(
\ker \widehat{\eta^\alpha} \cap \ker \widehat{\d \eta^\alpha}
\right) =
\{0\}
$.
\end{enumerate}
Provided that 
conditions (i) and (ii) in Definition \ref{kconman} hold,
condition (iii) is also equivalent to each one of these two conditions:
\begin{enumerate}
\item[\it (iii\,$''$)]
$\Tan M = \mathcal{D}^{\mathrm{C}} \oplus \mathcal{D}^{\mathrm{R}}$.
\item[\it (iii\,$'''$)]
$\Tan^*M = \mathcal{C}^{\mathrm{C}} \oplus \mathcal{C}^{\mathrm{R}}$.
\end{enumerate}
\end{obs}

\begin{obs}\rm
For the case $k=1$, 
a 1-contact structure is provided by a differential 1-form $\eta$,
and conditions in Definition \ref{kconman} 
mean the following:
(i) 
$\eta \neq 0$ at every point;
(iii) 
$\ker \widehat{\eta} \cap \ker \widehat{\d \eta} = \{0\}$,
which implies that 
$\ker \widehat{\d \eta}$ 
has rank 0 or~1;
(ii) 
means that
$\ker \widehat{\d \eta}$ 
has rank~1.
So, provided that (i) and (iii) hold,
condition (ii) is equivalent to saying that $\dim M$ is odd.
In this way, 
we recover the definition of contact structure.
\end{obs}

\begin{lem}
The Reeb distribution 
$\mathcal{D}^{\mathrm{R}}$ 
is involutive,
and therefore integrable.
\end{lem}
\begin{proof}
We use the relation
$$
\inn([X,X']) = \Lie_{X}\inn(X') -\inn(X') \Lie_{X} =
\d\inn(X)\inn(X')
+\inn(X) \d\inn(X')
-\inn(X')\d\inn(X)
-\inn(X')\inn(X) \d
\,. 
$$
When $X,X'$ are sections of $\mathcal{D}^{\mathrm{R}}$
and we apply this relation to the closed 2-form $\d \eta^\alpha$
the result is zero.
\end{proof}

\begin{thm}[Reeb vector fields]
\label{reebvf}
On a $k$-contact manifold $(M,\eta^\alpha)$ 
there exist $k$ vector fields 
$\Reeb_\alpha \in \mathfrak{X}(M)$,
the \emph{Reeb vector fields},
uniquely defined by the relations
\beq
\label{reebcontact}
\inn({\Reeb_\beta}) \eta^\alpha = \delta^\alpha_{\,\beta}
\,,\quad 
\inn({\Reeb_\beta}) \d\eta^\alpha = 0
\,.
\eeq
The Reeb vector fields commute:
$$
[\Reeb_\alpha,\Reeb_\beta] = 0
\,.
$$
In particular, 
${\cal D}^{\rm R} = 
\langle \Reeb_1,\ldots,\Reeb_k \rangle$.
\end{thm}
\begin{proof}
We take
$\Tan^*M = 
\mathcal{C}^{\mathrm{C}} \oplus \mathcal{C}^{\mathrm{R}}$.
The $\eta^\alpha$ are a global frame for the contact codistribution
$\mathcal{C}^{\mathrm{C}}$;
we can find a local frame $\eta^\mu$ for $\mathcal{C}^{\mathrm{R}}$.
So, $(\eta^\alpha;\eta^\mu)$ is a local frame for $\Tan^*M$.
The corresponding dual frame for $\Tan M$
is constituted by (smooth) vector fields
$(\Reeb_\beta;\Reeb_\nu)$,
where the $\Reeb_\beta$ are uniquely defined by
$$
    \langle \eta^\alpha , \Reeb_\beta \rangle =
     \delta^\alpha_{\:\beta}
    \,,\quad
    \langle \eta^\mu , \Reeb_\beta \rangle = 0 
    \,.
$$
Notice that the second set of relations
does not depend on the choice of the $\eta^\mu$
and simply means that
the $\Reeb_\beta$ are sections of 
$(\mathcal{C}^{\mathrm{R}})^\circ =
\mathcal{D}^{\mathrm{R}}$,
the Reeb distribution;
in other words, it means that, for every $\alpha$,
$\inn({\Reeb_\beta}) \d\eta^\alpha = 0$.
Notice finally that, since the $\eta^\alpha$ are globally defined, $\Reeb_\alpha$ also are.
    
To proof that 
the Reeb vector fields commute,
notice that
$$
\inn([\Reeb_\alpha,\Reeb_\beta]) \eta^\gamma = 0
\,,\quad
\inn([\Reeb_\alpha,\Reeb_\beta]) \d\eta^\gamma = 0
\,,
$$
which is a consequence of their definition
and of the above formula for 
$\inn([X,X'])$
when applied to them.
\end{proof}

\begin{prop}\label{prop-adapted-coord}
On a $k$-contact manifold there exist local coordinates
$(x^I;s^\alpha)$
such that
$$
    \Reeb_\alpha = \frac{\partial}{\partial s^\alpha}
    \,,\quad
    \eta^\alpha = \d s^\alpha - f_I^\alpha(x) \,\d x^I
    \,,
$$
where $f_I^\alpha(x)$ are functions depending only on the~$x^I$.
\end{prop}
\begin{proof}
The Reeb vector fields commute,
so there exist local coordinates 
$(x^I;s^\alpha)$
where
they can be straightened out simultaneously
(see for instance 
\cite[p.\,234]{Lee2013}):
$\displaystyle\Reeb_\alpha = \frac{\partial}{\partial s^\alpha}$.
\\
Now we express the contact forms in these coordinates.
First, relation
$
\inn(\Reeb_\beta) \eta^\alpha = \delta^\alpha_{\,\beta}
$
implies that
$
\eta^\alpha = \d s^\alpha - f_I^\alpha \,\d x^I
$,
where the functions $f_I^\alpha$ depend in principle 
on all the coordinates $(x^I;s^\alpha)$.
But then 
$\d \eta^\alpha = \d x^I \wedge \d f_I^\alpha$,
and the only way to ensure that
$\inn(\Reeb_\beta) \d \eta^\alpha = 0$
is that 
$\partial f_I^\alpha / \partial s^\beta = 0$.
\end{proof}

We will say that the coordinates provided by this proposition are
$\textbf{adapted}$
to the $k$-contact structure.

\begin{exmpl}
\label{example-canmodel}
\rm
(Canonical $k$-contact structure).
Given $k \geq 1$,
the manifold
$
M = (\oplus^k \Tan^*Q) \times \R^k
$
has a canonical $k$-contact structure defined by the 1-forms
$$
\eta^\alpha = \d s^\alpha - \theta^\alpha
\,,
$$
where
$s^\alpha$
is the $\alpha$-th cartesian coordinate of $\R^k$,
and 
$\theta^\alpha$
is the pull-back of the canonical 1-form of $\Tan^*Q$
with respect to the projection
$M \to \Tan^*Q$
to the $\alpha$-th direct summand.

Using coordinates $q^i$ on~$Q$ and natural coordinates
$(q^i,p_i^\alpha)$ on the $\alpha$-th copy of $\Tan^*Q$,
the local expressions of the contact forms are
$$
\eta^\alpha = \d s^\alpha - p^\alpha_i \,\d q^i
\,,
$$
from which 
$\d \eta^\alpha = \d q^i \wedge \d p_i^\alpha$,
the Reeb distribution is
$\mathcal{D}^{\mathrm{R}} = 
\langle 
\partial / \partial s^1 \ldots, \partial / \partial s^k
\rangle$,
and the Reeb vector fields are
$$
\Reeb_\alpha = \frac{\partial}{\partial s^\alpha}
\,.
$$
\end{exmpl}

\begin{exmpl}
\label{example2}
\rm
(Contactification of a $k$-symplectic manifold).
Let $(P,\omega^\alpha)$ be a $k$-symplectic manifold such that $\omega^\alpha=-\d\theta^\alpha$,
and consider $M = P \times \R^k$.
Denoting by $(s^1,\ldots,s^k)$ the cartesian coordinates of~$\R^k$,
and representing also by $\theta^\alpha$ the pull-back of~$\theta^\alpha$
to the product, we consider the 1-forms
$\eta^\alpha = \d s^\alpha - \theta^\alpha$ on~$M$.
Then $(M,\eta^\alpha)$ is a $k$-contact manifold because
$\mathcal{C}^{\mathrm{C}} = \langle \eta^1,\ldots,\eta^k \rangle$
has rank~k,
$\d \eta^\alpha = -\d\theta^\alpha$,
and
$\mathcal{D}^{\mathrm{R}}= \bigcap_\alpha\ker \widehat{\d\theta^\alpha}=
\langle \partial/\partial s^1,\ldots,\partial/\partial s^k\rangle$
has rank~k since $(P,\omega^\alpha)$ is $k$-symplectic,
and the last condition is immediate.

In particular, if $k=1$, let $P$ be a manifold endowed with a 1-form $\theta$,
and consider $M = P \times \R$.
Denoting by $s$ the cartesian coordinate of~$\R$,
and representing again by $\theta$ the pull-back of~$\theta$
to the product,
we consider the 1-form
$\eta = \d s - \theta$ on~$M$.
Then 
$\mathcal{C}^{\mathrm{C}} = \langle \eta \rangle$
has rank~1,
$\d \eta = -\d\theta$,
and
$\mathcal{D}^{\mathrm{R}} = \ker \widehat{\d\theta}$
has rank~1 if, and only if, $\d\theta$ is a symplectic form on~$P$.
In this case $M$ becomes a 1-contact manifold.
\end{exmpl}

\begin{exmpl}\rm
\label{example3}
Let $P = \R^6$ 
with coordinates
$(x,y,p,q,s,t)$.
The differential 1-forms
$$
\eta^1 = \d s - \frac12 ( y \d x - x \d y )
\quad ,\quad
\eta^2 = \d t - p \d x - q \d y
$$
define a 2-contact structure on~$P$.
Let us check the conditions of the definition.
First, the 1-forms are clearly linearly independent.
Then,
$$
\d \eta^1 = \d x \wedge \d y 
\quad ,\quad
\d \eta^2 = \d x \wedge \d p + \d y \wedge \d q
\ ,
$$
from which
$
\mathcal{D}^{\mathrm{R}} = 
\langle \partial/\partial s, \partial/\partial t \rangle
$,
which has rank~2.
Obviously none of these two vector fields belongs to the
kernel of the 1-forms, which is condition (iii). 
The Reeb vector fields are
$$
\Reeb_1 = \frac{\partial}{\partial s}
\quad ,\quad
\Reeb_2 = \frac{\partial}{\partial t}
\ .
$$
\end{exmpl}

\subsection{A Darboux theorem for \texorpdfstring{$k$}--contact manifolds}

The following result ensures the existence of canonical coordinates
for a particular kind of $k$-contact manifolds:

\begin{thm}[$k$-contact Darboux theorem]
\label{Darboux k-contact}
Let  $(M,\eta^\alpha)$ be a $k$--contact manifold of dimension $n+kn+k$
such that there exists an integrable subdistribution ${\cal V}$ of ${\cal D}^{\rm C}$
with ${\rm rank}\,{\cal V}=nk$.
Around every point of $M$, there exists a local chart of coordinates 
$(U;q^i,p^\alpha_i,s^\alpha)$, $1\leq\alpha\leq k \ ,\ 1\leq i \leq n$,
 such that
$$
\restr{\eta^\alpha}{U}=\d s^\alpha-p_i^\alpha\,\d q^i \;.
$$
In these coordinates,
$$
{\cal D}^{\rm R}\vert_{U}=\left\langle\Reeb_\alpha=\frac{\partial}{\partial s^\alpha}\right\rangle 
\,, \quad
{\cal V}\vert_{U}=\left\langle\frac{\partial}{\partial p_i^\alpha}\right\rangle
\ .
$$
These are the so-called \textbf{canonical} or \textbf{Darboux coordinates} of the $k$-contact manifold.
\end{thm}
\begin{proof}
{\bf (i)}: By Proposition \ref{prop-adapted-coord}, there exists a chart $(y^I;s^\alpha)$ of adapted coordinates around $p$ such that
$$
\Reeb_\alpha =  \frac{\partial}{\partial s^\alpha}
\,,\quad 
\eta^\alpha = \d s^\alpha-f^\alpha_I(y)\d y^I .
$$
Therefore, we can locally construct the quotient manifold 
$M/\mathcal{D}^{\mathrm{R}}\equiv\widetilde M$, with projection $\tilde\tau\colon M\to\widetilde M$,
and local coordinates $(\tilde y^I)$.

{\bf (ii)}: The distribution $\mathcal{D}^{\mathrm{C}}$, with ${\rm rank}\,\mathcal{D}^{\mathrm{C}}=nk+k$,
is $\bar\tau$-projectable because,
for every $\Reeb_\alpha\in\vf(\mathcal{D}^{\mathrm{R}})$, $Z\in\vf(\mathcal{D}^{\mathrm{C}})$
and $\d\eta^\beta$, we have
\begin{equation*}
\inn([\Reeb_\alpha,Z])\d\eta^\beta=
    \Lie_{\Reeb_\alpha}i(Z)\d\eta^\beta-\inn(Z)\Lie_{\Reeb_\alpha}\d\eta^\beta=
    -\inn(Z)\d\inn(\Reeb_\alpha)\eta^\beta=
    -\inn(Z)\d\delta_\alpha^\beta=0 \ ,
\end{equation*}
and so $[\Reeb_\alpha,Z]\in\vf(\mathcal{D}^{\mathrm{R}})$.

(Observe that this property is also a consequence of the condition (iii) in Theorem \ref{kconman}).

{\bf (iii)}:  For every $\beta$, the forms $\d\eta^\beta$
are  $\bar\tau$-projectable because,  by Theorem \ref{reebvf},
for every $\Reeb_\alpha\in\vf(\mathcal{D}^{\mathrm{R}})$, 
we have that $i(\Reeb_\alpha)d\eta^\beta=0$;
and hence
\beann
    \Lie_{\Reeb_\alpha}d\eta^\beta=\d \inn(\Reeb_\alpha)\eta^\beta=\d\delta_\alpha^\beta=0 \ .
\eeann
The $\tilde\tau$-projected forms $\tilde\omega^\beta\in\df^2(\widetilde M)$
such that $\d\eta^\beta=\tilde\tau^*\tilde\omega^\beta$
are obviously closed.
In coordinates they read $\tilde\omega^\beta=\d\tilde f^\beta_I(\tilde y)\wedge\d\tilde y^I$.

In addition, for every $Z,Y\in\Gamma(\mathcal{V})$ we have that,
as $\mathcal{V}$ is involutive,
$$
    \inn(Z)i(Y)\d\eta^\beta=\inn(Z)(\Lie_{Y}\eta^\beta-\d i(Y)\eta^\beta)=
    \inn(Z)\Lie_{Y}\eta^\beta=\Lie_{Y}\inn(Z)\eta^\beta-\inn([Y,Z])\eta^\beta=0 \ .
$$
Denoting by $\widetilde{\mathcal{V}}$ the distribution induced in $\widetilde{M}$
by ${\cal V}$ (which has ${\rm rank}\,\widetilde{\mathcal{V}}=kn$),
then, for every $\tilde Z,\tilde Y\in\Gamma(\widetilde{\mathcal{V}})$
if $Z,Y\in\Gamma(\mathcal{V})$ are such that 
$\tilde\tau_*Z=\tilde Z$, $\tilde\tau_*Y=\tilde Y$, we obtain that
\beq
    \label{primer}
    0=\inn(Z)\inn(Y)\d\eta^\beta=
    i(Z)i(Y)(\tilde\tau^*\tilde\omega^\beta)=
    \tilde\tau^*\inn(\tilde Z)i(\tilde Y)\tilde\omega^\beta \ ,
\eeq
and, as $\tilde\tau$ is a submersion, the map $\tilde\tau^*$ is injective and,
from \eqref{primer} we conclude that 
$\inn(\tilde Z)\inn(\tilde Y)\tilde\omega^\beta=0$.
(Observe that this proof is independent of the representative vector fields $Y,Z$ used,
because two of them differ in an element of $\ker\,\bar\tau_*=\Gamma(\mathcal{D}^{\mathrm{R}})$).
Thus we have proved that, for every $\beta$, we have
that $\tilde\omega^\beta\vert_{\widetilde{\mathcal{V}}\times\widetilde{\mathcal{V}}}=0$.

Finally, as a consequence of (ii), we have that 
$$
\ker\,\tilde\omega^1\cap\dots\cap\ker\,\tilde\omega^k=\{0\} \ ,
$$
Thus we conclude that
$(\widetilde M,\tilde\omega^\alpha,\widetilde{\mathcal{V}})$
is a $k$-symplectic manifold.

{\bf (iv)}: Therefore, by the Darboux theorem for $k$-symplectic manifolds \cite{Awane1992},
there are local charts of coordinates $(\tilde U;\tilde q^i,\tilde p_i^\alpha)$,
$1\leq i\leq n$, in $\widetilde M$,
 such that 
$$
\tilde\omega^\alpha\vert_{\tilde U}=\d\tilde q^i\wedge\d\tilde p^\alpha_i
\quad ; \quad
\widetilde{\mathcal{V}}\vert_{\tilde U}=
\left\langle\displaystyle\frac{\partial}{\partial\tilde p^\alpha_i}\right\rangle \ .
$$
Therefore, in $U=\tilde\tau^{-1}(\tilde U)\subset M$ we can take the coordinates 
$(y^I,s^\alpha)=(q^i,p^\alpha_i,s^\alpha)$, with
$q^i=\tilde q^i\circ\tilde\tau$ and
$p_i^\alpha=\tilde p_i^\alpha\circ\tilde\tau$
verifying the conditions of the theorem.
\end{proof}

This theorem allows us to consider the manifold presented in the example \ref{example-canmodel} as a canonical model
for these kinds of $k$-contact manifolds.
Furthermore, if $(M,\eta^\alpha)$ is a contactification of a $k$-symplectic manifold (example \ref{example2}), then there trivially exist Darboux coordinates.

\section{\texorpdfstring{$k$}--contact Hamiltonian systems}
\label{sect-hamsyst}

Using the geometric framework introduced in the previous section,
we are ready to deal with 
Hamiltonian systems with dissipation in field theories.

\begin{dfn}
\label{kconham}
A \textbf{$k$-contact Hamiltonian system} is a family $(M,\eta^\alpha,\H)$,
where $(M,\eta^\alpha)$ is a $k$-contact manifold,
and $\H\in \Cinfty(M)$ is called a \textbf{Hamiltonian function}.

The \textbf{$k$-contact Hamilton--De Donder--Weyl equations} for a map
$\psi \colon D\subset\R^k \to M$ is
\begin{equation}
\begin{cases}
i(\psi'_\alpha)\d\eta^\alpha = \big(\d\H - (\Lie_{\Reeb_\alpha}\H)\eta^\alpha\big)\circ\psi \ ,\\
i(\psi'_\alpha)\eta^\alpha = - \H\circ\psi \ .
\end{cases}
\label{hec}
\end{equation}
\end{dfn}

Let us express these equations in coordinates.
First, consider adapted coordinates $(x^I;s^\alpha)$,
with 
$R_\alpha = \partial / \partial s^\alpha$,
$\eta^\alpha = \d s^\alpha - f_I^\alpha(x) \,\d x^I$,
and
$\d \eta^\alpha = \frac12 \omega^\alpha_{IJ}\d x^I\wedge\d x^J$,
with
$\ds
\omega^\alpha_{IJ} = 
\frac{\partial f^\alpha_I}{\partial x^J} - 
\frac{\partial f^\alpha_J}{\partial x^I} 
$.
The map $\psi$ is expressed as
$\psi(t) = (x^I(t),s^\beta(t))$,
and
$\psi'_\alpha = 
(x^I,s^\beta; 
\partial x^I/\partial t^\alpha, \partial s^\beta/\partial t^\alpha
)$.
Then,
Hamilton--De Donder--Weyl equations
read
$$
\begin{cases}
\ds
\derpar{x^J}{t^\alpha} \, \omega_{JI}^\alpha =
\derpar{\H}{x^I} + \derpar{\H}{s^\alpha} f^\alpha_I \ ,
\\[2ex]
\ds
\derpar{s^\alpha}{t^\alpha} 
- f_I^\alpha \,\derpar{x^I}{t^\alpha} = -\H \ .
\end{cases}
$$
Analogously, in canonical coordinates, if 
$\psi=(q^i,p^\alpha_i,s^\alpha)$, 
these equations read
\beq
\begin{cases}
    \displaystyle \frac{\partial q^i}{\partial t^\alpha} = \frac{\partial\H}{\partial p^\alpha_i}\circ\psi \ ,
\\[15pt]
    \displaystyle \frac{\partial p^\alpha_i}{\partial t^\alpha} = 
    -\left(\frac{\partial\H}{\partial q^i}+ p_i^\alpha\frac{\partial\H}{\partial s^\alpha}\right)\circ\psi \ ,
\\[15pt]
    \displaystyle \frac{\partial s^\alpha}{\partial t^\alpha} = 
    \left(p_i^\alpha\frac{\partial\H}{\partial p^\alpha_i}-\H\right)\circ\psi\ ,
\end{cases}
\label{coor1}
\eeq

In order to give an alternative geometrical interpretation
we consider:

\begin{dfn}
\label{fieldeqk}
Let $(M,\eta^\alpha,\H)$ be a $k$-contact Hamiltonian system.
The \textbf{$k$-contact Hamilton--De Donder--Weyl equations} 
for a $k$-vector field ${\bf X}=(X_1,\dots,X_k)$ in $M$ are
\beq \begin{cases}
    \inn({X_\alpha})\d\eta^\alpha=\d\H-(\Lie_{\Reeb_\alpha}\H)\eta^\alpha \ ,\\
    \inn(X_\alpha)\eta^\alpha=-\H 
    \:.
    \end{cases}
    \label{fieldcontact}
\eeq
A $k$-vector field which is solution to these equations is called a
\textbf{Hamiltonian $k$-vector field}.
\end{dfn}

\begin{prop}
The $k$-contact Hamilton--De Donder--Weyl equations \eqref{fieldcontact} admit solutions.
They are not unique if $k>1$.
\end{prop}
\begin{proof}
A $k$-vector field $\mathbf{X}$ 
can be decomposed as
$\mathbf{X} =
\mathbf{X}^\mathrm{C} + 
\mathbf{X}^\mathrm{R} 
$
corresponding to the direct sum decomposition
$
\Tan M = \mathcal{D}^{\mathrm{C}} \oplus \mathcal{D}^{\mathrm{R}}
$.
If $\mathbf{X}$ is a solution to the 
$k$-contact Hamilton--De Donder--Weyl equations,
then 
$\mathbf{X}^\mathrm{C}$ 
is a solution to the first of these equations and
$\mathbf{X}^\mathrm{R}$ 
of the second one.

Now we introduce two vector bundle maps:
$$
\rho \colon \Tan M \to \oplus^k \,\Tan^*M
\,,\quad
\rho(v) = \left( \widehat{\d\eta^1}(v), \ldots,\widehat{\d\eta^k}(v) \right)
\,,
$$
$$
\tau \colon \oplus^k \Tan M \to \Tan^*M
\,,\quad
\tau(v_1\ldots,v_k) = \widehat{\d\eta^\alpha}(v_\alpha)
\,.
$$
Then, notice the following facts:
\begin{itemize}
\item 
$\ker \rho = \mathcal{D}^\mathrm{R}$ is the Reeb distribution.
\item
With the canonical identification $(E \oplus F)^* = E^* \oplus F^*$,
the transposed morphism of~$\tau$
is ${}^t \tau = -\rho$.
The proof uses that $^t  \widehat{\d\eta^\alpha} = -\widehat{\d\eta^\alpha}$.
\item
The first Hamilton--De Donder--Weyl equation for a $k$-vector field $\mathbf{X}$
can be written as
$\tau \circ \mathbf{X} = \d\H - (\Lie_{R_\alpha}\H) \eta^\alpha$.
\end{itemize}
A sufficient condition for this linear equation 
to have solutions $\mathbf{X}$
is that the right-hand side be in the image of~$\tau$,
that is,
to be annihilated by any section of 
$\ker {}^t\tau = \mathcal{D}^\mathrm{R}$.
But one easily computes that
$\ds
\inn(\Reeb_\beta)
(\d\H - (\Lie_{\Reeb_\alpha}\H) \eta^\alpha)
= 0$,
for any~$\beta$.
So we conclude that the first Hamilton--De Donder--Weyl equation
has solutions, 
and in particular solutions 
$\mathbf{X}^\mathrm{C}$ 
belonging to the contact distribution.

Finally, 
the second Hamilton--De Donder--Weyl equation for $\mathbf{X}$
admits solutions belonging to the Reeb distribution;
for instance,
$
\mathbf{X}^\mathrm{R} =
-\H \Reeb_1
$. 
Non-uniqueness is obvious.
\end{proof}

If ${\bf X}=(X_\alpha)$, 
is a $k$-vector field solution to equations \eqref{fieldcontact} and
$\displaystyle
X_\alpha= (X_\alpha)^\beta\derpar{}{s^\beta}+ (X_\alpha)^I\frac{\partial}{\partial x^I}$
is its expression in adapted coordinates of $M$,
then we have that
$$
\begin{cases}
\ds
(X_\alpha)^J \, \omega_{JI}^\alpha =
\derpar{\H}{x^I} + \derpar{\H}{s^\alpha} f^\alpha_I \,,
\\[2ex]
\ds
(X_\alpha)^\alpha
- f_I^\alpha \,(X_\alpha)^I= -\H \,.
\end{cases}
$$
In canonical coordinates, if
$\displaystyle
X_\alpha= (X_\alpha)^\beta\derpar{}{s^\beta}+ (X_\alpha)^i\frac{\partial}{\partial q^i}+
(X_\alpha)_i^\beta\frac{\partial}{\partial p_i^\beta}$, 
then
\beq
\begin{cases}
        \displaystyle(X_\alpha)^i= \frac{\partial\H}{\partial p^\alpha_i} \ ,\\[15pt]
        \displaystyle (X_\alpha)^\alpha_i = 
        -\left(\frac{\partial\H}{\partial q^i}+ p_i^\alpha\frac{\partial\H}{\partial s^\alpha}\right) \ ,\\[15pt]
        \displaystyle (X_\alpha)^\alpha = 
        p_i^\alpha\frac{\partial\H}{\partial p^\alpha_i}-\H\ ,
    \end{cases}
    \label{coor2}
\eeq

\begin{prop}
\label{prop1}
Let ${\bf X}$ be an \emph{integrable} $k$-vector field in $M$. 
Then
every integral section $\psi \colon D\subset\R^k \to M$ of ${\bf X}$
satisfies the $k$-contact equation (\ref{hec})
if, and only if,
${\bf X}$ is a solution to \eqref{fieldcontact}.
\end{prop}
\begin{proof}
This is a direct consequence of equations
\eqref{hec} and \eqref{fieldcontact},
and the fact that any point of~$M$
is in the image of an integral section of~$\mathbf{X}$.
\end{proof}

\begin{obs}
\label{remark1}
{\rm
As in the $k$-symplectic case,
equations (\ref{hec}) and (\ref{fieldcontact})
are not, in general, fully equivalent, since
a solution to \eqref{hec} may not be
an integral section of an integrable $k$-vector field solution to \eqref{fieldcontact}.
This remark will be relevant in Section
\ref{sect-symms}.

Furthermore,
in addition to not being unique,
solutions to equations (\ref{fieldcontact})
are not necessarily integrable.
}
\end{obs}

One can obtain the following alternative expression for the
Hamilton--De Donder--Weyl equations; its proof is immediate
(the case $k=1$ was done in
\cite{Bravetti2018})
\begin{prop}
The $k$-contact Hamilton--De Donder--Weyl equations \eqref{fieldcontact} are equivalent to 
\beq
\label{eq:kcontact3}
\begin{cases}
\Lie_{X_\alpha}\eta^\alpha = - (\Lie_{\Reeb_\alpha}\H)\eta^\alpha \:, \\
\inn(X_\alpha)\eta^\alpha = - \H 
\:.
\end{cases}
\eeq
\end{prop}

Finally, we present a sufficient condition for a $k$-vector field to be a solution of the Hamilton--De Donder--Weyl equations \eqref{fieldcontact}
without using the Reeb vector fields $\Reeb_\alpha$:

\begin{thm}
\label{eqset}
Let $(M,\eta^\alpha,\H)$ be a $k$-contact Hamiltonian system.
Consider the 2-forms $\Omega^\alpha= -\H\,\d\eta^\alpha+\d\H\wedge\eta^\alpha$. 
On the open set $O=\{p\in M;\H\neq 0\}$, if a $k$-vector field ${\bf X}=(X_\alpha)$ in $M$  verifies that
\begin{equation}
\label{hamilton-eqs-no-reeb}
    \begin{cases}
        i(X_\alpha)\Omega^\alpha = 0\,,\\
        i(X_\alpha)\eta^\alpha = -\H\,,
    \end{cases}
    \quad  \
\end{equation}
then ${\bf X}$ is a solution of the Hamilton--De Donder--Weyl equations \eqref{fieldcontact}).
\end{thm}
\begin{proof}
Suppose that ${\bf X}$ satisfies equations \eqref{hamilton-eqs-no-reeb}. Then,
\begin{equation*}
    0 = i(X_\alpha)\Omega^\alpha=-\H\,i(X_\alpha)\d\eta^\alpha+(i(X_\alpha)\d\H)\eta^\alpha+\H\,\d\H\ ,
\end{equation*}
and hence,
\begin{equation}\label{no-reeb-1}
    \H\,i(X_\alpha)\d\eta^\alpha=(i(X_\alpha)\d\H)\eta^\alpha+\H\,\d\H\ .
\end{equation}
Contracting this equation with every Reeb vector field $\Reeb_\beta$,
\begin{equation*}
    0 = \H\,i(\Reeb_\beta)i(X_\alpha)\d\eta^\alpha=
    (i(X_\alpha)\d\H)i(\Reeb_\beta)\eta^\alpha+\H\,i(\Reeb_\beta)\d\H=
    (i(X_\alpha)\d\H)\delta_\beta^\alpha+\H\,i(\Reeb_\beta)\d\H
    \ ,
\end{equation*}
and then  $i(X_\beta)\d\H = -\H\,i(\Reeb_\beta)\d\H$,
for every $\beta$. Using this in equation \eqref{no-reeb-1}, we get
\begin{equation*}
    \H i(X_\alpha)\d\eta^\alpha=\H(\d\H -(i(\Reeb_\alpha)\d\H)\eta^\alpha)=
    \H(\d\H -(\Reeb_\alpha(\H))\eta^\alpha)\ ;
\end{equation*}
and therefore $i(X_\alpha)\d\eta^\alpha=\d\H-(\Reeb_\alpha(\H))\eta^\alpha$.
\end{proof}

Bearing in mind Definition \ref{kconham}
and Proposition \ref{prop1}, 
we can write the equations
for the corresponding integral sections:

\begin{prop}
On the open set $O=\{p\in M;\H\neq 0\}$, if $\psi\colon D\subset\R^k\to O$ is an integral section of a $k$-vector field solution
to equations \eqref{hamilton-eqs-no-reeb}, then it is a solution to
\begin{equation}
 \begin{cases}
i(\psi'_\alpha)\Omega^\alpha = 0 
\:,\\
i(\psi'_\alpha)\eta^\alpha = - \H\circ\psi 
\:.
\end{cases}
\label{hec2}
\end{equation}
\end{prop}

\section{Examples}
\label{examples}

\subsection{Damped vibrating string}
\label{subsec-string}

It is well known that a vibrating string can be described
within the Lagrangian formalism.
Let us use coordinates $(t,x)$ for the time and the space,
and let $u$ be the separation of a point of the string 
from the equilibrium position;
we also denote by $u_t$ and $u_x$ the derivatives of~$u$
with respect to the independent variables.
Let $\rho$ be the linear mass density of the string
and $\tau$ its tension (they are assumed to be constant).
Taking as Lagrangian density
$$
L = \frac{1}{\,2\,} \rho u_t^2 - \frac{1}{\,2\,} \tau u_x^2
$$
and defining  
$c^2 = \tau/\rho$
one obtains as the Euler--Lagrange equation the wave equation
$$
u_{tt} = c^2 u_{xx} 
\ .
$$

We rather need to express this equation within the Hamiltonian formalism.
We add the momenta of~$u$ as dependent variables $p^t$, $p^x$.
The Legendre transformation $\mathcal{F}\L$ of~$\L$ is such that
$$
\mathcal{F} L^*(p^t) = \rho \,u_t 
\quad ,\quad
\mathcal{F} L^*(p^x) = -\tau \,u_x \,,
$$
and from it we obtain the Hamiltonian function
$$ 
H = \frac{1}{\,2\rho\,} (p^t)^2 - \frac{1}{\,2\tau\,} (p^x)^2 \ .
$$
As we have a scalar field $u$ and two independent variables $(t,x)$,
this corresponds to a $2$-symplectic theory in the canonical model $\oplus^2 \Tan^*\R$.
The Hamilton--De Donder--Weyl equations are
$$
\frac{\partial u}{\partial t} = 
\frac{\partial H}{\partial p^t} \,,\quad
\frac{\partial u}{\partial x} = 
\frac{\partial H}{\partial p^x} \,,\quad
\frac{\partial p^t}{\partial t} + \frac{\partial p^x}{\partial x} =
-\frac{\partial H}{\partial u} \ .
$$
For our Hamiltonian, they read
$$
\frac{\partial u}{\partial t} = 
\frac{p^t}{\rho} 
\,,\quad
\frac{\partial u}{\partial x} = 
-\frac{p^x}{t} 
\,,\quad
\frac{\partial p^t}{\partial t} + \frac{\partial p^x}{\partial x} =
0 \ .
$$
The last equation yields immediately
$\displaystyle
\rho \frac{\partial^2 u}{\partial t^2} 
- \tau \frac{\partial^2 u}{\partial x^2}
= 0
$,
which is the wave equation.

A simple model of a vibrating string with an external damping 
can be obtained by adding to the wave equation
a dissipation term proportional to the speed of an element of the string.
So this is given by the equation
$$
\frac{\partial^2 u}{\partial t^2} 
- c^2 \frac{\partial^2 u}{\partial x^2}
+ k \frac{\partial u}{\partial t}
= 0 \ ,
$$
where $k > 0$ is the damping constant
\cite[p.\,284]{Salsa2015}.
Now we will show that this equation can be formulated as a contact Hamiltonian system.
To this end, according to example \ref{example2}, we add two additional variables $s^t$ and~$s^x$,
and define an extended Hamiltonian
$\H(u,p^t,p^x,s^t,s^x)$
by
$$
\H = H + h
\ ,
$$
where
$\displaystyle
H = \frac{1}{2\rho} (p^t)^2 - \frac{1}{2\tau} (p^x)^2
$
is the Hamiltonian of the undamped vibrating string
and
$$
h = ks^t \ .
$$
Then the contact Hamilton--De Donder--Weyl equations for $H$ read 
$$ 
\begin{cases}
\displaystyle
\frac{\partial u}{\partial t} = \frac1\rho p^t \ ,
\\[1ex]
\displaystyle\frac{\partial u}{\partial x} = -\frac1\tau p^x \ ,
\\[1ex]
\displaystyle\frac{\partial p^t}{\partial t} + \frac{\partial p^x}{\partial x} = 
-k p^t \ ,
\\[1ex]
\displaystyle\frac{\partial s^t}{\partial t} + \frac{\partial s^x}{\partial x} =
\frac{1}{2\rho} (p^t)^2 - \frac{1}{2\tau} (p^x)^2 - k s^t \:.
\end{cases}
$$
Using the first and second equations within the third we obtain
$$
\rho \frac{\partial^2 u}{\partial t^2}
- \tau \frac{\partial^2 u}{\partial x^2 }
+ k\rho \frac{\partial u}{\partial t}
= 0 \:,
$$
which is the equation of the damped string.

\subsection{Burgers' Equation}
\label{subsec-Burgers}

Burgers' equation (Bateman, 1915) is a remarkable nonlinear partial differential equation
that appears in several areas of applied mathematics.
There is one dependent variable $u$ and two independent variables $(t,x)$, and reads
$$
u_t + u u_x = k \,u_{xx}
\ ,
$$
where $k \geq 0$ is a diffusion coefficient
\cite[p.\,217]{Salsa2015}.
Notice that it looks quite similar to the heat equation
$
u_t = k \,u_{xx}
$;
indeed we will show that Burgers' equation can be formulated 
as a contactification of the heat equation.
This will be performed in several steps.

\paragraph{Lagrangian formulation of the heat equation}

We will need a Hamiltonian formulation of the heat equation.
This can be obtained from a Lagrangian formulation of it.
Although the heat equation is not variational,
it can be made variational
by considering an auxiliary dependent variable~$v$,
and taking as Lagrangian
---see for instance~\cite{Ibragimov2004}
$$
L = - k \, u_x v_x 
- \frac12 \left( v u_t - u v_t \right)
\,,
$$
whose Euler--Lagrange equations are
$$
[L]_u =
k \,v_{xx} + v_t = 0
\;,\quad
[L]_v =
k \,u_{xx} - u_t = 0
\,.
$$
The first equation is linear homogeneous, 
therefore always has solutions (for instance $v=0$).
So,
there is a correspondence between 
solutions of the heat equation
and solutions of the Euler--Lagrange equations of~$L$ with $v=0$.

\paragraph{Hamiltonian formulation of the heat equation}

Now we apply the Donder--Weyl Hamiltonian formalism to~$L$.
Its Legendre map (fiber derivative) is a map
$
\mathcal{F} L \colon 
\oplus^2 \Tan \R^2 \to 
P = \oplus^2 \Tan^*\R^2 
$.
The phase space is
$P \approx \R^6$,
where the fields are $u$, $v$ 
and their respective momenta
$p^t,p^x$ and~$q^t,q^x$ 
with respect to the independent variables.
The Legendre map $\mathcal{F} \L$ of~$\L$ relates
these momenta with the configuration fields and their velocities:
$$
\mathcal{F} \L^*(p^x) = - k\,v_x 
\:,\quad
\mathcal{F} \L^*(p^t) = - \frac12 \,v \,,
$$
$$
\mathcal{F} \L^*(q^x) = - k\,u_x 
\:,\quad
\mathcal{F} \L^*(q^t) =  \frac12 \,u \,.
$$
So the image $P_0 \subset P$ of the Legendre map 
is defined by the two constraints
$$
p^t \approx - \frac12 v
\:,\quad
q^t \approx   \frac12 u
\,.
$$
We will use coordinates $(u,v,p^x,q^x)$ on it.
Finally, the Hamiltonian function on $P_0$ is
$$ 
H = - \frac1k \,p^x q^x \,.
$$

The manifold $P$ is endowed with an exact 2-symplectic structure
defined by the canonical 1-forms
$$
p^t \d u + q^t \d v 
\:,\quad
p^x \d u + q^x \d v \:.
$$
Their pullbacks to $P_0$ 
are not anymore a 2-symplectic structure
(under the standard definition),
but nevertheless we have 
two differential 1-forms
$$
\theta^t = \frac{1}{2} (-v \d u + u \d v) 
\:,\quad
\theta^x = p^x \d u + q^x \d v 
$$
with 
$$
\omega^t = -\d \theta^t = 
- \d u \wedge \d v 
\:,\quad
\omega^x = -\d \theta^x = 
\d u \wedge \d p^x + \d v \wedge \d q^x 
\,.
$$
Now, 
let $\psi \colon \R^2 \to P_0$ be a map,
$\psi = (u,v,p^x,q^x)$.
It is readily computed that 
the Hamilton--De Donder--Weyl equation
$$
\inn({\psi_t'}) \omega^t +
\inn({\psi_x'}) \omega^x =
\d H \circ \psi
$$
for $\psi$ is equivalent to
$$
\partial_t v - \partial_x p^x = 0
\,,\
- \partial_t u - \partial_x q^x = 0
\,,\
\partial_x u = -  \frac1k \,q^x
\,,\
\partial_x v = - \frac1k \,p^x
\,.
$$
Using the last two equations in the first two ones,
we obtain the heat equation for~$u$,
and its complementary equation for~$v$:
$$
\partial_t u = k \,\partial_x^2 u
\,,\quad
\partial_t v = -k \,\partial_x^2 v
\,.
$$
Notice again that the equation for~$v$ is linear homogeneous, 
therefore there is a bijection 
between solutions of this system with $v=0$,
and solutions of the heat equation.
So this is the Hamiltonian formulation of the heat equation we sought.

\paragraph{Contact Hamiltonian formulation of the Burgers' equation}

Now we take again the manifold $P_0$ 
and its two differential 1-forms 
to construct a 2-contact structure.
To this end 
we consider the product manifold
$M = P_0 \times \R^2 = \R^6$,
with the cartesian coordinates $s^t,s^x$
of $\R^2$, 
and construct the contact forms
$$
\eta^t = \d s^t - \theta^t
\:,\quad 
\eta^x = \d s^x - \theta^x 
\:,
$$
where we keep the same notation for 
$\theta^t, \theta^x$
as 1-forms on~$M$.
Their differentials are the same 2-forms
$\omega^t, \omega^x$
written before.
With the notations of section~3,
since $\eta^t,\eta^x$ are linearly independent 
at each point,
$
\mathcal{C}^\mathrm{C} =
\langle \eta^t,\eta^x \rangle
$
is a regular codistribution of rank~2,
and  
$
\mathcal{D}^\mathrm{R} =
\langle \Reeb_t,\Reeb_x \rangle
$,
with
$\Reeb_t = \partial/\partial s^t$
and
$\Reeb_x = \partial/\partial s^x$,
is a regular distribution of rank~2.
Moreover,
$
\mathcal{D}^\mathrm{C} \cap \mathcal{D}^\mathrm{R}
= \{0\}
$,
since no nonzero linear combination of the
$\partial/\partial s^t, \partial/\partial s^x$
is annihilated by the contact forms.
Therefore,
$(M;\eta^t,\eta^x)$
is a 2-contact manifold.
This is indeed example~\ref{example3}.

Finally, 
we take as a contact Hamiltonian the function
$$ 
\H = - \frac1k p^x q^x + \gamma u s^x 
.
$$
In this way we obtain a 2-contact Hamiltonian system
$(M,\eta^\alpha,\H)$.
\\
Let us compute the contact Hamilton--De Donder--Weyl equations for this system,
\begin{align*}
\inn({\psi_t'})\omega^t +
\inn({\psi_x'})\omega^x 
&=
\d \H 
- (\Lie_{\Reeb_t}\H) \, \eta^t
- (\Lie_{\Reeb_x}\H) \, \eta^x 
\,,
\\
\inn({\psi_t'}) \eta^t +
\inn({\psi_x'}) \eta^x &=
-\H
\,.
\end{align*}
The first equation is similar to the
contactless Hamilton--De Donder--Weyl equation,
with just some additional terms:
$$
\partial_t v - \partial_x p^x = 
\gamma (s^x + u p^x)
\,,\
- \partial_t u - \partial_x q^x = 
 \gamma u q^x
\,,\
\partial_x u = - \frac1k \,q^x
\,,\
\partial_x v = -\frac 1k \,p^x
\,.
$$
Again, putting the latter two equations
inside the former ones,
we obtain
$$
\partial_t u -\gamma k \,u \,\partial_x u = 
k \,\partial_x^2 u
\,,\quad
\partial_t v + \gamma k  \,u \,\partial_x v = 
-k \,\partial_x^2 v + \gamma s^x
\,.
$$
Setting the value of the constant $\gamma = -1/k$,
the first equation is Burgers' equation for~$u$.
\\
Finally, 
the second contact Hamilton--De Donder--Weyl equation reads:
$$
\partial_t s^t -
\frac{1}{2} ( -v \, \partial_t u + u \, \partial_t v )
+
\partial_x s^x
- p^x \,\partial_x u - q^x \, \partial_x v = 
\frac1k \,p^x q^x - \gamma \, u s^x
\,.
$$
Again, notice that these equations admit solutions 
$(u,v,p^x,q^x,s^t,q^t)$
with
$v,p^x,s^t,s^x = 0$,
$q^x = -k \,\partial_x u$,
and $u$ an arbitrary solution of Burgers' equation.
Therefore, we conclude that the Burgers' equation
can be described by the 2-contact Hamiltonian system
$(M,\eta^\alpha,\H)$.

\section{Symmetries and dissipation laws}
\label{sect-symms}

Finally, 
we introduce some basic ideas about symmetries and 
deduce an associated dissipation law 
for $k$-contact Hamiltonian field theories.

\subsection{Symmetries}

There are different concepts of symmetry of a problem, 
depending on which structure was preserved. 
One can put the emphasis on 
the transformations that preserve the geometric structures of the problem,
or on the transformations that preserve its solutions
\cite{Gracia2002}.
This has been done in particular for the $k$-symplectic
Hamiltonian formalism
\cite{Roman2007}.
We will apply the same idea to $k$-contact Hamiltonian systems.
First we consider symmetries as those transformations
that map solutions of the equations into other solutions.
So we define:

\begin{dfn}
\label{def-dynsym}
Let  $(M,\eta^\alpha,\H)$ be a $k$-contact Hamiltonian system.
\begin{itemize}
\item 
A \textbf{dynamical symmetry} is a diffeomorphism $\Phi:M\rightarrow M$ such that, 
for every solution $\psi$ to the $k$-contact Hamilton--De Donder--Weyl equations \eqref{fieldcontact}, 
$\Phi\circ\psi$ is also a solution.
\item 
An \textbf{infinitesimal dynamical symmetry} is a vector field $Y\in\mathfrak{X}(M)$ whose local flow is made of local dynamical symmetries.
\end{itemize}
\end{dfn}

We will give a characterization of symmetries in terms of $k$-vector fields.
First let us recall the following fact about integral sections:

\begin{lem}
Let $\Phi\colon M\to M$ be a diffeomorphism and $\mathbf{X}=(X_1,\dots,X_k)$ 
a $k$-vector field in $M$. If $\psi$ is an
integral section of $\mathbf{X}$, then $\Phi\circ\psi$ is an integral section of $\Phi_*\mathbf{X}=(\Phi_*X_\alpha)$.
In particular, if ${\bf X}$ is integrable then
$\Phi_*{\bf X}$ is also integrable.
\end{lem}

Then we have:

\begin{prop}
\label{prop-sym1}
If $\Phi \colon M \rightarrow M$ is a dynamical symmetry then, 
for every integrable $k$-vector field ${\bf X}$ solution to 
the $k$-contact Hamilton--De Donder--Weyl equations \eqref{fieldcontact}, 
$\Phi_*{\bf X}$ is another solution.

On the other side, if $\Phi$ transforms every $k$-vector field ${\bf X}$ 
solution to equations \eqref{fieldcontact} into another solution, 
then for every integral section $\psi$ of ${\bf X}$, we have that 
$\Phi\circ\psi$ is
a solution to the $k$-contact Hamilton--De Donder--Weyl equations \eqref{hec}.
\end{prop}
\begin{proof}
$( \Rightarrow )$ 
Let $x\in M$ and let $\psi$ be an
integral section of ${\bf X}$ through the point $\Phi^{-1}(x)$, that is $\psi(t_o)=\Phi^{-1}(x)$.
We know that
$\psi$ is a solution to the $k$-contact Hamilton--De Donder--Weyl equations \eqref{hec} and,
as $\Phi$ is a dynamical symmetry, $\Phi\circ \psi$ is also a solution
to equations (\ref{hec}). 
But, by the preceding lemma,
it is an integral section of $\Phi_*{\bf X}$ trough
the point $\Phi(\psi(t_o))=\Phi(\Phi^{-1}(x))=x$
and
hence we have that $\Phi_* {\bf X}$  must be a solution
to equations \eqref{fieldcontact} at the
points $(\Phi\circ\psi)(t)$ and,
in particular, at the point $(\Phi\circ\psi)(t_o)=x $.

$( \Leftarrow )$ 
On the other side, let
${\bf X}$ be a solution to equations \eqref{fieldcontact} and
$\psi\colon U\subset\R^k\to M$ an
integral section of ${\bf X}$.
Then,  by hypothesis, $\Phi_*{\bf X}$ is also a solution and then,
as a consequence of the previous lemma, 
we have that $\Phi\circ\psi$
is a solution to the $k$-contact Hamilton--De Donder--Weyl equations (\ref{hec}).
\end{proof}

In geometrical physics, 
among the most relevant symmetries 
there are those that let the geometric structures invariant:

\begin{dfn}
\label{def-hamconsym}
Let  $(M,\eta^\alpha,\H)$ be a $k$-contact Hamiltonian system.

\noindent
A \textbf{Hamiltonian $k$-contact symmetry} is a diffeomorphism $\Phi\colon M\rightarrow M$ such that
$$
\Phi^*\eta^\alpha=\eta^\alpha
\:,\quad 
\Phi^*\mathcal{H}=\mathcal{H}
\:.
$$
An \textbf{infinitesimal Hamiltonian $k$-contact symmetry} is a vector field $Y\in \X(M)$ whose local flow is a local Hamiltonian $k$-contact symmetry; that is,
$$
\Lie_Y\eta^\alpha=0
\:,\quad 
\Lie_Y\mathcal{H}=0 
\:.
$$
\end{dfn}

\begin{prop}
\label{prop-sym2}
Every (infinitesimal) Hamiltonian $k$-contact symmetry preserves the Reeb vector fields, that is; $\Phi_*\Reeb_\alpha=\Reeb_\alpha$ (or $[Y,\Reeb_\alpha]=0$).
\end{prop}
\begin{proof}
We have that
\beann
\inn(\Phi^{-1}_*\Reeb_\alpha)(\Phi^*\d\eta^\alpha)&=&
\Phi^*\inn(\Reeb_\alpha)\d\eta^\alpha=0 \:, \\
\inn(\Phi^{-1}_*\Reeb_\alpha)(\Phi^*\eta^\alpha)&=&
\Phi^*\inn(\Reeb_\alpha)\eta^\alpha=1 \:,
\eeann
and, as $\Phi^*\eta^\alpha=\eta^\alpha$ and the Reeb vector fields are unique, from these equalities
we conclude that $\Phi_*\Reeb_\alpha=\Reeb_\alpha$.

The proof for the infinitesimal case is immediate from the definition.
\end{proof}

Finally, as a consequence of these results, we obtain the relation between Hamiltonian $k$-contact symmetries and dynamical symmetries:

\begin{prop}
\label{prop-sym3}
(Infinitesimal) Hamiltonian $k$-contact symmetries are (infinitesimal) dynamical symmetries.
\end{prop}
\begin{proof} Let $\psi$ be a solution to the $k$-contact-De Donder--Weyl equations \eqref{hec}, and $\Phi$ a Hamiltonian $k$-contact symmetry. Then
\beann{}
\inn((\Phi\circ\psi)'_\alpha)\eta^\alpha &=&
\inn(\Phi_*(\psi'_\alpha))((\Phi^{-1})^*\eta^\alpha)=
(\Phi^{-1})^*\inn(\psi'_\alpha)\eta^\alpha
\\ &=& (\Phi^{-1})^*(-\H\circ\psi)=
-\H\circ(\Phi\circ\psi) \ ,
\\
\inn((\Phi\circ\psi)'_\alpha)\d\eta^\alpha&=&
\inn(\Phi_*(\psi'_\alpha))((\Phi^{-1})^*\d\eta^\alpha)=
(\Phi^{-1})^*\inn(\psi'_\alpha)\d\eta^\alpha
\\
&=&(\Phi^{-1})^*\Big((\d\H - (\Lie_{\Reeb_\alpha}\H)\eta^\alpha)\circ\psi\Big)
\\
&=&\Big(\d(\Phi^{-1})^*\H - (\Lie_{(\Phi^{-1})^*\Reeb_\alpha}(\Phi^{-1})^*\H)(\Phi^{-1})^*\eta^\alpha\Big)\circ(\Phi\circ\psi)
\\
&=&\Big(\d\H - (\Lie_{\Reeb_\alpha}\H)\eta^\alpha\Big)\circ(\Phi\circ\psi) \ .
\eeann{}
The proof for the infinitesimal case is immediate from the definition.
\end{proof}

\subsection{Dissipation laws}

In many mechanical systems without dissipation, we are interested in quantities which are conserved along a solution. 
Classical examples are the energy or the different momenta. From a physical point of view, if a system has dissipation, these quantities are not conserved. 
This behavior is explicitly shown for Hamiltonian contact systems in the so called energy dissipation theorem 
\cite{Lainz2018} 
which says that,
if $X_\H$ is a contact Hamiltonian vector field, then
$$
\Lie_{X_\H} \H = -(\Lie_{\Reeb}\H)\H .
$$
This equation shows that, in a contact system, 
the dissipations are exponential with rate $-(\Lie_{\Reeb}\H)$. 
In dissipative field theories, a similar structure can be observed in the first equation of \eqref{eq:kcontact3}, 
which can be interpreted as the dissipation of the contact forms $(\eta^\alpha)$.
Then, bearing in mind the definition of conservation law for field theories
as stated in
\cite{Olver1986},
and taking into account the Remark \ref{remark1},
this suggests the following definitions of dissipation laws for $k$-contact Hamiltonian systems:

\begin{dfn}
\label{def-disip}
Let $(M,\eta^\alpha,\mathcal{H})$ be a $k$-contact Hamiltonian system.
A map 
$F\colon M\rightarrow\R^k$,
$F=(F^1,\dots,F^k)$, is said to satisfy:
\ben
\item
The \textbf{dissipation law for sections} if, for every solution $\psi$ to the $k$-contact Hamilton--De Donder--Weyl equations (\ref{hec}), 
the divergence of 
$
F \circ \psi = (F^\alpha\circ\psi) \colon \R^k \to \R^k
$,
which is defined as usual by
$\ds
\mathrm{div} (F \circ \psi) = 
{\partial (F^\alpha {\circ}\, \psi)}/{\partial t^\alpha}
$,
satisfies that
\begin{equation}
\label{cons-law}
\mathrm{div} (F \circ \psi) =
- \left[\strut (\Lie_{\Reeb_\alpha}\H) F^\alpha \right] \circ \psi
\ .
\end{equation}
    \item 
The \textbf{dissipation law for $k$-vector fields} if, for every $k$-vector field ${\bf X}$  solution to the $k$-contact Hamilton--De Donder--Weyl equations (\ref{fieldcontact}),
the following equation holds:
\beq
\label{cons-law field}
\Lie_{X_\alpha}F^\alpha=-(\Lie_{\Reeb_\alpha}\H)F^\alpha\ .
\eeq
\end{enumerate}
\end{dfn}

Both concepts are partially related by the following property:

\begin{prop}
\label{prop-disip}
If $F=(F^\alpha)$ satisfies the dissipation law for sections then,  
for every integrable $k$-vector field ${\bf X}=(X_\alpha)$
which is a solution to the $k$-contact Hamilton--De Donder--Weyl equations \eqref{fieldcontact}, 
we have that the equation
\eqref{cons-law field}
holds for ${\bf X}$.

On the other side, if \eqref{cons-law field} holds for a $k$-vector field ${\bf X}$, then
\eqref{cons-law} holds for every integral section $\psi$ of ${\bf X}$.
\end{prop}
\begin{proof}
If $F=(F^\alpha)$ satisfies the dissipation law for sections,
${\bf X}=(X_\alpha)$ is an integrable $k$-vector field
which is a solution to the $k$-contact Hamilton--De Donder--Weyl equations \eqref{fieldcontact}, 
and $\psi\colon\R^k\to M$ an integral section of ${\bf X}$, then
by Proposition \ref{prop1}, $\psi$ is a solution to the
$k$-contact Hamilton--De Donder--Weyl equations \eqref{hec};
therefore 
$$
(\Lie_{X_\alpha}F^\alpha) \circ \psi =
\frac{\d}{\d t^\alpha} (F^\alpha \circ \psi) =
\mathrm{div} (F \circ \psi)=
- \left[\strut (\Lie_{\Reeb_\alpha}\H) F^\alpha \right] \circ \psi
\ ,
$$
and, as ${\bf X}$ is integrable, there exists an integral section through every point,
hence the result follows.

On the other side, if \eqref{cons-law field} holds, then the statement is a straightforward consequence of the above expression.
\end{proof}

\begin{lem}
\label{braket}
If $Y$ is an infinitesimal dynamical symmetry then, for every solution ${\bf X}=(X_\alpha)$
to the $k$-contact Hamilton--De Donder--Weyl equations \eqref{fieldcontact}, we have that
$$
\inn([Y,X_\alpha])\eta^\alpha=0 
\:,\quad
\inn([Y,X_\alpha])\d\eta^\alpha=0 
\:.
$$
\end{lem}
\begin{proof}
Let $F_\varepsilon$ be the local $1$-parameter groups of diffeomorphisms generated by $Y$. 
As $Y$ is an infinitesimal dynamical symmetry, 
$
\inn(F_{\varepsilon}^*  X_\alpha)\eta^\alpha  =
\inn(X_\alpha)\eta^\alpha
$, 
because both are solutions to the Hamilton--De Donder--Weyl equations \eqref{fieldcontact}. 
Then, as the contraction is a continuous operation,
\begin{equation*}
\inn([Y,X_\alpha])\eta^\alpha
=
\inn \left( \lim_{\varepsilon \to 0} 
\frac{F_{\varepsilon}^* X_\alpha - X_\alpha}{\varepsilon}
\right)\eta^\alpha
=
\lim_{\varepsilon \to 0}
\frac{\inn(F_{\varepsilon}^*X_\alpha)\eta^\alpha  - \inn(X_\alpha)\eta^\alpha}{\varepsilon}
=
0\:.
\end{equation*}
The equation involving $\d \eta^\alpha$ is proved in the same way.
\end{proof}

Then we have the following fundamental result which associates dissipation laws for $k$-vector fields with symmetries:

\begin{thm}[Dissipation theorem]
\label{theor-disip}
If $Y$ is an infinitesimal dynamical symmetry, 
then $F^\alpha=-i(Y)\eta^\alpha$ 
satisfies the dissipation law for $k$-vector fields.
\end{thm}
\begin{proof}
Consider a solution $X$ to the 
$k$-contact Hamilton--De Donder--Weyl equations \eqref{fieldcontact}.
From Lemma \ref{braket}, 
we have that 
$\inn([Y,X_\alpha])\eta^\alpha=0$, 
therefore
$$
    \Lie_{X_\alpha} (i(Y)\eta^\alpha) =
    \inn([X_\alpha,Y])\eta^\alpha+\inn(Y)\Lie_{X_\alpha}\eta^\alpha
    =
    -(\Lie_{\Reeb_\alpha}\H)\inn(Y)\eta^\alpha
    \,.
$$
\end{proof}

\subsection{Examples}

\subsubsection{Mechanics: energy dissipation}

In this case $k=1$. Let $X_h$ be the Hamiltonian contact vector field. 
Then, as $[X_h,X_h]=0$,
the vector field $X_h$ is a dynamical symmetry. 
Then, applying the dissipation theorem we have that  $-\inn(X_h)\eta=\H$
satisfies the dissipation law
$$
\Lie_{X_h}\H=-(\Lie_{\Reeb}\H)\H \ ,
$$
which is the energy dissipation theorem \cite{Lainz2018}.

\subsubsection{Damped vibrating string}

We resume the example discussed in Section~\ref{subsec-string}.
The vector field
$\displaystyle\frac{\partial}{\partial q}$
is a contact symmetry. Hence
it induces the map
$$
F=\Big( 
-\inn\Big(\frac{\partial}{\partial q}\Big)\eta^t, 
-\inn\Big(\frac{\partial}{\partial q}\Big)\eta^x
\Big)
=
(p^t,p^x) \;,
$$
which follows the dissipation law
 \eqref{cons-law field}:
$$
\Lie_{X_t}p^t+\Lie_{X_x}p^x =
-(\Lie_{\Reeb_t}\H)p^t-(\Lie_{\Reeb_x}\H)p^x=-2\gamma p^t \,.
$$
Over a solution $(X_t,X_x)$, this law is
$$
p^t_t+p^x_x=-2\gamma p^t \,.
$$

\subsubsection{Burgers' equation}

Now we take up the example discussed in Section~\ref{subsec-Burgers}.

The vector field $\displaystyle\frac{\partial}{\partial v}$ is not a contact symmetry
but a dynamical symmetry.
Hence
it induces the map
$$
F=\Big(
-\inn\Big(\frac{\partial}{\partial v}\Big)\eta^t,
-\inn\Big(\frac{\partial}{\partial v}\Big)\eta^x
\Big) =
\Big(\frac{1}{2k}u,p^x_v\Big)\ .
$$
which follows the dissipation law
 \eqref{cons-law field}:
$$
\Lie_{X_t}\frac{1}{2k}u+\Lie_{X_x}p^x_v=
-(\Lie_{\Reeb_t}\H)\frac{1}{2k}u-(\Lie_{\Reeb_x}\H)p^x_v \ .
$$
Over a solution $(X_t,X_x)$ this law is
$$
\frac{1}{2k} \frac{\partial u}{\partial t}+
\frac{\partial p^x_v}{\partial x} =
\gamma u p^x_v \ ,
$$
which is the Burgers' equation again.

\subsubsection{A model of two coupled vibrating strings with damping}

Consider $M=\oplus^2 \Tan^*\mathbb{R}^2 \times \R^2$, 
with coordinates $(t,x,q^1,q^2,p^t_1,p^t_2,p^x_1,p^x_2,s^t,s^x)$, 
where $q^1$ and $q^2$ represent the displacements of both strings.
When it is endowed with the forms
\begin{eqnarray*}
\eta^t =\d s^t-p^t_1\d q^1-p^t_2\d q^2
\,,\quad
\eta^x =\d s^x-p^x_1\d q^1-p^x_2\d q^2 
\,,
\end{eqnarray*}
we have the $2$-contact manifold $(M,\eta^t,\eta^x)$. Now consider the Hamiltonian function
$$
\H=
\frac12((p^t_1)^2+(p^t_2)^2+(p^x_1)^2+(p^x_2)^2)+
G(z)+\gamma s^t \:,
$$
where $G$ is a function that represents a coupling of the two strings, 
and which we assume to depend only on 
$z = \left( (q^2)^2+(q^1)^2 \right)^{1/2}$.

A simple computation shows that the following vector field is an infinitesimal contact symmetry
$$
Y=q^1\derpar{}{q^2}-q^2\derpar{}{q^1}+p^t_1\derpar{}{p^t_2}-p^t_2\derpar{}{p^t_1}+p^x_1\derpar{}{p^x_2}-p^x_2\derpar{}{p^x_1}\ ,
$$
and it induces the map
$$
F=(-\inn(Y)\eta^t,-\inn(Y)\eta^x)=(q^1p^t_2-q^2p^t_1,q^1p^x_2-q^2p^x_1)\ ,
$$
which satisfies the dissipation equation along a solution $(X^t,X^x)$
\begin{align*}
    \Lie_{X^x}(q^1p^t_2-q^2p^t_1)+\Lie_{X^t}(q^1p^x_2-q^2p^x_1)&=
    q^1\Big(\derpar{p^t_{2}}{t}+\derpar{p^x_{2}}{x}\Big)-q^2
    \Big(\derpar{p^t_{1}}{t}+\derpar{p^x_{1}}{x}\Big)\\
    &=-\gamma(q^1p^t_2-q^2p^t_1)\ .
\end{align*}

\section{Conclusions and outlook}

In this paper we have introduced a Hamiltonian formalism
for field theories with dissipation.
Using techniques from contact geometry and 
the $k$-symplectic Hamiltonian formalism,
we have developed a new geometric framework,
defining the concepts of $k$-contact manifold
and $k$-contact Hamiltonian system.
In the same way that a contact structure allows to describe dissipative mechanics,
a $k$-contact structure gives a transparent description of dissipative
systems in field theory
over a $k$-dimensional parameter space.
It is important to stress that,
to our knowledge, this is the first time 
these geometric structures are presented.

In more detail, we have stated the definition of $k$-contact structure on a manifold,
as a family of $k$ differential 1-forms
satisfying certain properties
(Definition~\ref{kconman}).
This implies the existence of two special tangent distributions,
in particular the Reeb distribution,
which is spanned by $k$ Reeb vector fields.
We have proved the existence of special systems of coordinates,
and a Darboux-type theorem for a particular kind of these manifolds.

Using this structure,
the notion of $k$-contact Hamiltonian system is defined.
The corresponding field equations 
(Definition~\ref{kconham})
are a generalization of both 
the contact Hamilton equations of dissipative mechanics,
and 
the Hamilton--De Donder--Weyl equations of Hamiltonian field theory.

We have analyzed the concept of symmetry for dissipative Hamiltonian field theories.
We study two natural types of symmetries:
those preserving the solutions to the field equations, 
and those preserving the geometric $k$-contact structure and the Hamiltonian function.
We have also defined the notion of dissipation law
in order to extend 
the energy dissipation theorem of contact mechanics,
stating a dissipation theorem
which relates symmetries and dissipation laws 
which is analogous to the conservation theorems 
in the case of conservative field theories.

Two relevant examples 
show the significance of our framework:
the vibrating string with damping, 
and the Burgers equation.
In our presentation, Burgers' equation is obtained as a
contactification of the heat equation;
to this end, we have also provided, for the first time,
a Hamiltonian field theory describing the heat equation.

The results of this work open some future lines of research.
The first one would be the definition of
the Lagrangian formalism for dissipative field theories 
and the associated Hamiltonian formalism.
The case of a singular Lagrangian
would require to define the notion of
\emph{$k$-precontact structure} on a manifold,
that is,
a family of $k$ 1-forms $\eta^\alpha$
that do not meet all the conditions of Definition~\ref{kconman}.
(The case $k=1$ has been recently analyzed in \cite{DeLeon2019}.)
Another interesting issue would be to deepen the study of symmetries for $k$-contact Lagrangian and Hamiltonian systems.

\subsection*{Acknowledgments}
We acknowledge the financial support from the 
Spanish Ministerio de Econom\'{\i}a y Competitividad
project MTM2014--54855--P, 
the Ministerio de Ciencia, Innovaci\'on y Universidades project
PGC2018-098265-B-C33,
and the Secretary of University and Research of the Ministry of Business and Knowledge of
the Catalan Government project
2017--SGR--932.

\bibliographystyle{abbrv}
\addcontentsline{toc}{section}{References}
\itemsep 0pt plus 1pt
\small

\end{document}